\documentclass[twocolumn,pra,aps]{revtex4}
\usepackage{amssymb}
\usepackage{amsmath}
\usepackage{amsthm}
\usepackage{hyperref}

 \newcommand{\ket}[1]{\left|{#1}\right\rangle}

\newtheorem{prop}{Proposition}

\newtheorem{theorem}[prop]{Theorem}
\theoremstyle{remark}

\theoremstyle{definition}

\begin{document}

\title{Efficient semiquantum key distribution}

\author{Ming-Ming Wang}
\email{bluess1982@126.com}
\affiliation{Shaanxi Key Laboratory of Clothing Intelligence, School of Computer Science, Xi'an Polytechnic University, Xi'an 710048, China}
\affiliation{State and Local Joint Engineering Research Center for Advanced Networking \& Intelligent Information Services, School of Computer Science, Xi'an Polytechnic University, Xi'an 710048, China}
\author{Lin-Ming Gong}
\affiliation{Shaanxi Key Laboratory of Clothing Intelligence, School of Computer Science, Xi'an Polytechnic University, Xi'an 710048, China}
\affiliation{State and Local Joint Engineering Research Center for Advanced Networking \& Intelligent Information Services, School of Computer Science, Xi'an Polytechnic University, Xi'an 710048, China}
\author{Lian-He Shao}
\affiliation{Shaanxi Key Laboratory of Clothing Intelligence, School of Computer Science, Xi'an Polytechnic University, Xi'an 710048, China}
\affiliation{State and Local Joint Engineering Research Center for Advanced Networking \& Intelligent Information Services, School of Computer Science, Xi'an Polytechnic University, Xi'an 710048, China}

\date{\today}

\begin{abstract}
Quantum cryptography has attracted much attention in recent years. In most existing quantum cryptographic protocols, players usually need the full quantum power of generating, manipulating or measuring quantum states.
Semiquantum cryptography was proposed to deal with the issue that some players require only partial quantum power, such as preparing or measuring quantum states in the classical basis, which simplifies the implementations of quantum cryptography.
However, the efficiency of the existing semiquantum cryptographic protocols was relatively low from a practical point of view.
In this paper, we devise some new semiquantum key distribution (SQKD) protocols which highly improve the efficiency of the most well-known SQKD protocols [Phys. Rev. Lett. 99, 140501 (2007) \& Phys. Rev. A 79, 052312 (2009)].
By letting players select their actions asymmetrically, the efficiency of our new protocols can be made asymptotically close to 100\%.
Besides, one of our proposed protocols also utilizes the discarded $X$-SIFT bits in the original SQKD protocol, which further improves the efficiency of SQKD.
We prove that the proposed SQKD protocols are completely robust against the most general attack.
\end{abstract}

\maketitle

\section{Introduction}
Quantum physics has greatly enriched our abilities of computing \cite{Shor97,Grover-225} and communications \cite{BB84,E91}. As one of its application, quantum cryptographic \cite{GisinRibordy-370} has been studied extensively both in theory and in experiment. Conventionally, a quantum cryptographic protocol, such as quantum key distribution \cite{BB84,E91}, requires all of the players equipped with quantum capabilities of preparing, manipulating or measuring qubits. However, it is expensive and inconvenient for all players equipped with full quantum capabilities. In 2007, Boyer et al. \cite{BoyerKenigsberg-1125} proposed a novel idea of quantum key distribution, where one of the player Alice has full quantum capabilities, while the other player Bob is classical. The ``classical" Bob either measures the qubits Alice sent in the classical basis (the $Z$ basis) and resends it in the same state he found, or reflects the qubits without any change.
They called the protocol as  ``quantum key distribution with classical Bob'' or ``semiquantum key distribution (SQKD)''.
Their first SQKD protocol \cite{BoyerKenigsberg-1125} used four quantum states, each of which is randomly prepared in the $Z$ or $X$ basis. The idea was further extended in \cite{BoyerGelles-1126}, where two similar protocols were presented based on measurement-resend and randomization.
Lu and Cai \cite{LuCai-1136} presented a SQKD protocol with classical Alice.
In 2009, Zou et al. \cite{ZouQiu-1127} pointed out that the original SQKD protocol \cite{BoyerKenigsberg-1125} can be simplified by employing less than four states and they proposed five different SQKD protocols.
In 2011, Wang et al. \cite{Wang-1138} proposed an efficient SQKD protocol based on entangled states.
In 2015, Zou et al. \cite{ZouQiu-1124} proposed a SQKD protocol in which the ``classical'' player does not need the measurement capability, and just needs preparing, sending and reordering qubits.
Krawec \cite{Krawec-1134} proposed a mediated SQKD protocol where two ``classical'' players can establish a secret key with the help of a quantum server.
In 2017, Boyer et al. \cite{BoyerKatz-1140} proposed an experimentally feasible SQKD protocol by using a ``controllable mirror.'' And this SQKD protocol is further simplified in \cite{BoyerLiss-1141}.
Recently, Krawec \cite{Krawec-1142} introduced a new SQKD protocol against certain practical attacks.
Zhang et al. \cite{ZhangQiu-1143} proved the security of the single-state SQKD in \cite{ZouQiu-1127} from information theory aspect.
Furthermore, several multiuser SQKD protocols were put forward \cite{LuCai-1136,ZhangGong-1137,Krawec-1134}.
Besides SQKD, other semiquantum cryptographic protocols, such as semiquantum secret sharing \cite{LiChan-1119,WangZhang-1129,Lvzhou-1118,XieLi-1120}, semiquantum information splitting \cite{NieLi-1130}, and semiquantum secure direct communication \cite{ZouQiu-1135} have also been studied.

Security is first requirement for a quantum cryptographic protocol. For proving the security of a QKD protocol, an important step is to show that the protocol is robust \cite{BoyerKenigsberg-1125}. A QKD protocol is robust if any adversarial attempt to learn some information on the key inevitably induces some errors. Boyer et al. \cite{BoyerKenigsberg-1125} divided robustness into three classes: \textit{completely robust}, \textit{partly robust}, and \textit{completely nonrobust}.
Clearly, completely robust protocols are securer than partly robust protocols; partly robust protocols could still be secure, but completely nonrobust protocols are automatically proven insecure.
Boyer et al. \cite{BoyerKenigsberg-1125,BoyerGelles-1126} proved that their protocol is completely robust. Besides, some others SQKD protocols were showed to be completely robust \cite{ZouQiu-1127,ZouQiu-1124,Wang-1138,ZhangGong-1137,Walter-1149,ZhangQiu-1143}.

For practical QKD protocols, another important factor is efficiency, which can be defined as the proportion of the final sifted bits (sifted key) to the whole bits generated by quantum carriers. In the original SQKD protocol \cite{BoyerKenigsberg-1125}, the quantum qubits used are divided into three sets, i.e. the INFO set for generating the final sifted key, the TEST set for confining the error rate on INFO bits and the CTRL set for bounding Eve's information; the INFO set and the TEST set together are called SIFT set. Note that in most of the existing SQKD protocols, the qubits efficiency is relatively low.
For example, in the well-known SQKD protocols in \cite{BoyerKenigsberg-1125,BoyerGelles-1126,ZouQiu-1127}, the ratio of the INFO qubits for generating final sifted key is less than $\frac{1}{8}$ to the total qubits in most of the protocols, while the best efficiency is the protocol 4 in \cite{ZouQiu-1127} which is less than $\frac{1}{4}$. The reason behind the low efficiency of these protocols lies in the fact that the parameters of these protocols are fixed, i.e. Bob's choice for selecting his action either SIFT or CTRL is totally random (with equal probability $\frac{1}{2}$), which is not a good way for improving the efficiency of a SQKD protocol.

To improve the qubits efficiency, we propose some novel asymmetric SQKD protocols following the ideas proposed by Lo et al. \cite{LoChau-1148} for the BB84 QKD protocol.
By asymmetric, we mean that the actions performed by each player are selected with unequal probabilities.
Different from previous SQKD protocols \cite{BoyerKenigsberg-1125,BoyerGelles-1126,ZouQiu-1127}, in our SQKD protocols, Bob chooses his two actions randomly and independently with Alice's choice, but not uniformly. In other words, Bob's SIFT and CTRL actions are chosen with \emph{substantially different} probabilities \cite{LoChau-1148}.
Our protocols differ from the measurement-resend SQKD protocol in Refs. \cite{BoyerKenigsberg-1125,BoyerGelles-1126} in two ways. Firstly, we use different parameters to control the numbers of the INFO, TEST and CTRL bits that used for generating sifted key and detecting eavesdropping. As Alice and Bob are now much more likely to generate INFO bits, the fraction of discarded data is greatly reduced, thus achieving a significant gain in efficiency.
Secondly, the $X$-SIFT bits discarded in \cite{BoyerKenigsberg-1125,BoyerGelles-1126} are now used to generate SIFT bits in one of our new protocols, which further improves the efficiency of the protocol.

The rest of the paper is organized as follows.
In Sect. 2, our improved SQKD protocols are demonstrated. We first show our SQKD protocol 1 using four
different qubits, which is an improved version of the original SQKD protocol. Then, we devise an novel SQKD protocol utilizing the $X$-SIFT bits rather than discarding them. We also present our SQKD protocol 3 using only $\ket{+}$.
And the proof of complete robustness of each protocol is given immediately after the description the protocol.
This paper is further discussed and concluded in Sect. 3.

\section{Efficient asymmetric SQKD protocols without entanglement}

\subsection{Protocol 1: Asymmetric SQKD}

The following is our SQKD protocol using single-qubits as information carriers, which is an improved asymmetric version of the original SQKD \cite{BoyerKenigsberg-1125} where a quantum Alice can generate a secret key with a classical Bob. Our asymmetric SQKD protocol 1 performs as follows.

 (1)   Alice generates a random string $a \in \{0, 1\}^N$ containing nearly $\gamma_1 N$ bits of 0, with $ \frac{1}{2} < \gamma_1 < 1$. She generates and sends $\ket{\phi}_1, \ket{\phi}_2, \cdots, \ket{\phi}_N$ to classical Bob, where $\ket{\phi}_i$ is selected from $\{ \ket{0}$,   $\ket{1} \}$ if $a_i = 0$, or from $ \{\ket{+} = \frac{1}{\sqrt{2}} (\ket{0}+\ket{1}), \ket{-} = \frac{1}{\sqrt{2}} (\ket{0}-\ket{1}) \}$  if $a_i = 1$.

  (2)  Bob generates a random string $b = \{0, 1\}^N$ containing nearly $\gamma_2 N$ bits of 0, with $ \frac{1}{2} < \gamma_2 < 1$.
  When the $i$th qubit comes, he measures it in the $Z$ basis and resends it in the same state he found (SIFT it) if  $b_i = 0$, or reflects it back without any modification (CTRL it) if $b_i = 1$.

  (3) Alice measures each CTRL qubit in the basis she sent it.

  (4) Alice publishes $a$ and Bob publishes $b$.

  (5) Alice checks the error rate on the CTRL bits. Alice and Bob abort the protocol if the error rate is higher than the predefined threshold $P_t$.

  (6) Alice and Bob randomly select $\xi \gamma_1 \gamma_2 N $ measurement results of the SIFT bits to be the TEST bits, where
  $0<\xi<\frac{1}{2}$ is the proportion of TEST bits to the SIFT bits. They check the error rate on the TEST bits. Alice and Bob abort the protocol if the error rate is higher than $P_t$.

  (7) Alice and Bob set the remaining bits of SIFT bits as the INFO bits.

  (8) Alice and Bob perform classical post-processing procedures, such as information reconciliation and privacy amplification, to produce the final key from the INFO bits, similar to the BB84 QKD protocol \cite{GisinRibordy-370}

\begin{theorem}
The SQKD protocol 1 is completely robust.
\end{theorem}

\begin{proof}
Note that our SQKD protocol 1 use the same method as is used in the original measure-resend SQKD \cite{BoyerGelles-1126} to detect Eve's attack. The only difference between these two protocols lies in the proportion of the CTRL and TEST bits used for eavesdropping detection.
The proof of complete robustness of the SQKD protocol 1 can be concluded directly from the Theorem 3 in  \cite{BoyerGelles-1126}.
\end{proof}

\subsection{Protocol 2: Asymmetric SQKD using $X$-SIFT}

Note that the $X$-SIFT bits have been discarded in \cite{BoyerKenigsberg-1125}. In our SQKD protocol 2, we utilizes the $X$-SIFT bits for generating TEST or INFO bits. Our SQKD protocol 2 performs as follows.

 (1)   Alice generates a uniformly random string $a \in \{0, 1\}^N$, where $N = (\kappa+\tau+\lambda)(1+\delta)$ with $(\kappa+\tau)$ is the length of the SIFT bits that includes the INFO bits with length $\kappa$ and the TEST bits with length $\tau$, and $\lambda$ is the length of the CTRL string, while $\delta > 0 $ is a fixed parameter that is the same as the original SQKD protocol \cite{BoyerKenigsberg-1125}.
 Alice generates and sends $\ket{\phi}_1, \ket{\phi}_2, \cdots, \ket{\phi}_N$, where $\ket{\phi}_i$ is selected from $\{ \ket{0}$,   $\ket{1} \}$ if $a_i = 0$, or from $ \{\ket{+} = \frac{1}{\sqrt{2}} (\ket{0}+\ket{1}), \ket{-} = \frac{1}{\sqrt{2}} (\ket{0}-\ket{1}) \}$  if $a_i = 1$. Alice sends a qubit only after receiving the previous one from Bob.

  (2)  Bob generates a random string $b = \{0, 1\}^N$ containing nearly $(\kappa+\tau)(1+\delta)$ bits of 0.
  When the $i$th qubit comes, he measures it in the $Z$ basis and resends it in the same state he found (SIFT it) if  $b_i = 0$, or reflects it back without any modification (CTRL it) if $b_i = 1$.

  (3) Alice uses an $N$-qubit register to save all qubits coming back from Bob.

  (4) Bob announces $b$ after Alice receives the last qubit.

  (5) Alice measures each CTRL qubit in the basis she sent it and measures each SIFT qubit in the $Z$ basis. Then Alice checks the error rate on the CTRL bits. Alice and Bob abort the protocol if the error rate is higher than the predefined threshold $P_t$.

  (6) Alice and Bob randomly select $\tau$ measurement results of the SIFT bits to be the TEST bits. They check the error rate on the TEST bits. Alice and Bob abort the protocol if the error rate is higher than $P_t$.

  (7) Alice and Bob select the first $\kappa$ remaining bits of SIFT bits as the INFO bits.

  (8) Alice and Bob perform classical post-processing procedures, such as information reconciliation and privacy amplification, to produce the final key from the INFO bits, similar to the BB84 QKD protocol \cite{GisinRibordy-370}

\begin{theorem}
The SQKD protocol 2 is completely robust.
\end{theorem}

\begin{proof}
Note that our SQKD protocol 2 is similar to the SQKD protocol 2 in Ref. \cite{ZouQiu-1127} since Alice in both protocols only sends one-qubit each time.
The proof of complete robustness of the SQKD protocol 2 in \cite{ZouQiu-1127} (also see \cite{BoyerMor-1147,ZouQiu-1146}) lies on the usage of $\ket{+}$ and $\ket{0}$, which implies the complete robustness of our SQKD protocol 2.
\end{proof}

\subsection{Protocol 3: Asymmetric SQKD using only $\ket{+}$}

In Ref. \cite{ZouQiu-1127}, Zou et al. also proposed an interesting SQKD protocol with only $\ket{+}$ qubits. Here, we simply modify this protocol to an asymmetric one with high efficiency.
Our SQKD protocol 4 is performed as follows.

 (1)   Alice sets $N = (\kappa+\tau+\lambda)(1+\delta)$, with the parameters have the same meaning as our Protocol 2. Alice generates and sends $N$ qubits $\ket{+}^{\otimes N}$. Alice sends a qubit only after receiving the previous one from Bob.

  Steps (2)-(8) are the same as our Protocol 2.

\begin{theorem}
The SQKD protocol 3 is completely robust.
\end{theorem}

\begin{proof}
The complete robustness of our SQKD protocol 3 can be conclude directly from the proof of complete robustness of the SQKD protocol 4 in \cite{ZouQiu-1127} (also see \cite{BoyerMor-1147,ZouQiu-1146}).
\end{proof}

\section{Discussions and Conclusions}

A key point of asymmetric SQKD lies in the fact that the security of SQKD protocols can be guaranteed
 if we could set the numbers of the CTRL and TEST bits to reach some certain thresholds, and it unnecessary to keep a fixed ratio of these test bits with the increase of total generated classical bits.
The proportion of different types of bits generated in our SQKD protocol 1 is described in Table \ref{tab:1}, while our SQKD protocols 2 \& 3 are in Table \ref{tab:2}. The efficiency of our new protocols can be made asymptotically close to 100\%, i.e. these SQKD protocols can be made arbitrarily efficient.

\begin{table}
\caption{Proportion of bits used in Protocol 1, where $\gamma$ is the probability that Alice prepares states in the $Z$ basis and Bob chooses to SIFT, while $\xi$ is proportion the TEST bits to the SIFT bits.}
\label{tab:1}
\begin{center}
\begin{tabular}{lll}
\hline\noalign{\smallskip}
 Name & Proportion &  Usage\\\hline
 $Z$-SIFT & ~$\gamma_1\gamma_2 $  & INFO \& TEST bits \\
                   &~~~ $\xi \gamma_1\gamma_2 $ & ~~~INFO  bits \\
                   &~~~ $(1-\xi) \gamma_1\gamma_2 $ &~~~TEST bits  \\
 $X$-SIFT & ~$(1-\gamma_1)\gamma_2 $&Discard \\
 $Z$-CTRL &~$\gamma_1(1-\gamma_2) $ & CTRL bits\\
 $X$-CTRL &~ $(1-\gamma_1)(1-\gamma_2) $          & CTRL bits \\
\hline\noalign{\smallskip}
\end{tabular}
\end{center}
\end{table}

\begin{table}
\caption{Proportion of bits used in Protocol 2 and Protocol 3, where $\kappa$ is the number of INFO bits, $\tau$ is the number of TEST bits and $\lambda$ is the number of the CTRL bits.}
\label{tab:2}
\begin{tabular}{lc}
\hline\noalign{\smallskip}
   Name & Proportion\\
\noalign{\smallskip}\hline \noalign{\smallskip}
CTRL &  $\approx \frac{\lambda}{\kappa+\tau+\lambda}$ \\
SIFT  &  $\approx \frac{\kappa+\tau}{\kappa+\tau+\lambda}$ \\
~~TEST   &  $\approx \frac{\tau}{\kappa+\tau+\lambda}$ \\
~~INFO   &  $\approx \frac{\kappa}{\kappa+\tau+\lambda}$ \\
\noalign{\smallskip}\hline
\end{tabular}
\end{table}

Besides, in the original SQKD protocol \cite{BoyerKenigsberg-1125,BoyerGelles-1126}, nearly $\frac{1}{4}$ of all generated bits, the $X$-SIFT bits, have been discarded. In our SQKD protocol 2, we make these $X$-SIFT bits usable for generating SIFT bits. We use three different parameter $\kappa$, $\tau$ and $\lambda$ to control the bits ratios used for the INFO, TEST and CTRL bits.
If we let $\kappa=\frac{N}{4}$, $\tau=\frac{N}{4}$ and $\lambda=\frac{N}{2}$, our protocol 2 is similar to the measure-resend SQKD protocol in \cite{BoyerKenigsberg-1125,BoyerGelles-1126} except that $X$-SIFT bits has been used, while if we set $\kappa > \frac{N}{4}$, our protocol 2 becomes even more efficient.
Similar results can also be obtained from the other two protocols.

In summary, we have proposed some improved SQKD protocols against the most well-known SQKD protocols in \cite{BoyerKenigsberg-1125,BoyerGelles-1126,ZouQiu-1124} with efficiency of these improved protocols asymptotically close to 100\% based on the idea of asymmetric QKD \cite{LoChau-1148}.
Firstly, we introduce asymmetric SQKD protocols with players Alice and Bob perform their actions in a random but non uniform manner, which makes most of the qubits be used for generating sifted key, but the others for generating tested key is sufficient to guarantee the security of the protocol.
Secondly, we show that the $X$-SIFT qubits discarded in previous SQKD protocols can be used for generating sifted key and tested key.
It is clear that the asymmetric method can also be used directly to other existing SQKD protocols.
We hope this research would be helpful for practical implementations of semiquantum cryptographic protocols.

\section*{Acknowledgements}
\noindent
This project was supported by the National Natural Science Foundation of China (Grant No. 61601358).

\end{document}